\documentclass[twocolumn]{article}

\usepackage{color}
\usepackage[draft]{todonotes}
\usepackage{caption}
\usepackage{xspace}
\usepackage{amsmath}
\usepackage{amsthm}
\usepackage{graphicx}
\usepackage{cleveref}
\usepackage[hyphens]{url}
\usepackage{algorithm}
\usepackage{algorithmic}
\usepackage{authblk}

\graphicspath{{./figures/}}

\newtheorem{definition}{Definition}
\newtheorem{theorem}{Theorem}

\author[1]{Mustafa Al-Bassam}
\affil[1]{Department of Computer Science\authorcr University College London\authorcr \protect\url{m.albassam@cs.ucl.ac.uk}}
\date{}

\begin{document}
\newcommand\mustafa[1]{}
\newcommand\alberto[1]{}

\newcommand{\publication}[1]{\noindent\vspace*{-1em}\raisebox{20pc}[0pt][0pt]{\hspace*{-0pt}\noindent\parbox[t]{6.5in}{\sl{#1}}}}

\newcommand\sysname{LazyLedger\xspace}

\newcommand{\cf}{cf.\@\xspace}
\newcommand{\vs}{vs.\@\xspace}
\newcommand{\etc}{etc.\@\xspace}
\newcommand{\ala}{ala\@\xspace}
\newcommand{\wrt}{w.r.t.\@\xspace}
\newcommand{\etal}{\textit{et al.}\@\xspace}
\newcommand{\eg}{\textit{e.g.,}\@\xspace}
\newcommand{\ie}{\textit{i.e.,}\@\xspace}
\def\first{({\it i})\xspace}
\def\second{({\it ii})\xspace}
\def\third{({\it iii})\xspace}
\def\fourth{({\it iv})\xspace}
\def\fifth{({\it v})\xspace}
\def\sixth{({\it vi})\xspace}

    \title{\sysname: A Distributed Data Availability Ledger With Client-Side Smart Contracts}
    \maketitle
    \publication{This document is a draft; feedback is welcome.}

    \begin{abstract}
        We propose \sysname, a design for distributed ledgers where the blockchain is optimised for solely ordering and guaranteeing the availability of transaction data. Responsibility for executing and validating transactions is shifted to only the clients that have an interest in specific transactions relating to blockchain applications that they use. As the core function of the consensus system of a distributed ledger is to order transactions and ensure their availability, consensus participants do not necessarily need to be concerned with the contents of those transactions. This reduces the problem of block verification to data availability verification, which can be achieved probabilistically with sub-linear complexity, without downloading the whole block. The amount of resources required to reach consensus can thus be minimised, as transaction validity rules can be decoupled from consensus rules. We also implement and evaluate several example \sysname applications, and validate that the workload of clients of specific applications does not significantly increase when the workload of other applications that use the same chain increase.
    \end{abstract}

    \section{Introduction}
So far, blockchain-based distributed ledger platforms such as Bitcoin \cite{nakamoto2008} and Ethereum \cite{buterin2013} have adopted similar consensus design paradigms, where the validity of the blocks proposed by block producers is determined by \first whether it is the block producer's turn to propose a block and \second whether the transactions in the block are valid according to some state machine. Traditional consensus protocols such as Practical Byzantine Fault Tolerance \cite{castro1999} have also taken a similar approach, where consensus nodes (replicas) process transactions according to a state machine.

The scalability issues that have plagued decentralised blockchains \cite{croman2016} can be attributed to the fact that in order to run a node that validates the blockchain, the node must download, process and validate every transaction included in the chain. As a result, various scalability efforts have emerged including on-chain scaling via sharding \cite{al-bassam2018sharding, kokoris-kogias2018}, which aims to split the state of the blockchain into multiple shards so that transactions can be processed by different consensus groups in parallel, and off-chain scaling via state channels \cite{poon2016, miller2017}, which takes the approach of moving transactions off-chain and using the blockchain as a settlement layer.
    
However, it is also worth exploring alternative blockchain design paradigms that may be suitable for different types of applications, where nodes that need to validate the blockchain in order to determine the correct chain do not need to validate the contents of the blocks. Instead, the end-users of applications that store information on the blockchain can be concerned with the validation of such contents. This would remove the bottleneck where nodes need to validate everyone else's transactions, and reducing the problem of validating the blockchain to simply verifying that the contents of the block are available (the data availability problem \cite{al-bassam2018}), so that end-users can meaningfully access the information needed to apply state transitions on their applications. In such a paradigm, the blockchain is used solely for ordering and making available messages, rather than executing and verifying the state machine transitions of transactions. Because messages for applications are executed by end-users off-chain, the logic of these applications does not need to be defined on-chain, and thus application logic can be written in any programming language or environment, and changing the logic does not require a hard-fork of the chain. \mustafa{Sell this more by referencing DAO hack etc?}

A result of reducing blockchain validation to the data availability problem is that one can fully achieve consensus on new messages without downloading the entire set of messages, using probabilistic data availability verification techniques \cite{al-bassam2018}, as consensus participants do not need to process the contents of messages.

Philosophically, \sysname can be thought of as a system of `virtual' sidechains \cite{back2014} that live on the same chain, in the sense that transactions associated with each application only need to be processed by users of those applications, similar to the fact that only users of a specific sidechain need to process transactions of that sidechain. However, because all applications in \sysname share the same chain, the availability of the data of all their transactions are equally and uniformly guaranteed by the same consensus group, unlike in traditional sidechains where each sidechain may have a different (smaller) consensus group.

In this paper, we make the following contributions:

\begin{itemize}
    \item We design a blockchain, \sysname, where consensus and transaction validity is decoupled, and describe two alternative block validity rules which just ensure that block data is available. One is a simple rule where nodes simply download the blocks themselves, and the other is probabilistic but more efficient as nodes do not need to download entire blocks.
    \item We build an application-layer on top of our proposed blockchain, where end-user clients can efficiently query the network for data relating only to their applications, and only need to execute transactions related to their applications.
    \item We implement and evaluate several example \sysname applications; including a currency, a name registrar and a petitions system.
\end{itemize}

\noindent\textbf{Outline}: \Cref{sec:background} presents a background of the technical concepts \sysname relies on; \Cref{sec:model} presents the \sysname threat model, node types, and blockchain model; \Cref{sec:block-validity-rule} presents the block validity rules of the \sysname blockchain; \Cref{sec:application-model} presents the \sysname application model and how applications can be built on top of its blockchain; \Cref{sec:implementation} presents an evaluation of a prototype of \sysname and some example applications; \Cref{sec:related-work} presents a comparison with related work; and \Cref{sec:conclusion} concludes.
    \section{Background}\label{sec:background}

\subsection{Blockchains}

The data structure of a blockchain consists of a chain of blocks. Each block contains two components: a header and a list of transactions. In addition to other metadata, the header stores at minimum the hash of the previous block (thus enabling the chain property), and the root of the Merkle tree that consists of all transactions in the block.
    
Blockchain networks implement a consensus algorithm \cite{bano2017} to determine which chain should be favoured in the event of a fork, \eg if proof-of-work \cite{nakamoto2008} is used, then the chain with the most accumulated work is favoured. They also have a set of transaction validity rules that dictate which transactions are valid, and thus blocks that contain invalid transactions will never be favoured by the consensus algorthim and should in fact always be rejected.
    
Full nodes (also known as `fully validating nodes') are nodes which download both the block headers as well as the list of transactions, verifying that all transactions are valid according to some transaction validity rules. This is necessary in order to know which blocks have been accepted by the consensus algorithm.
    
There are also `light' clients which only download block headers, and assume that the list of transactions are valid according to the transaction validity rules. These nodes verify blocks against the consensus rules, but not the transaction validity rules, and thus assume that the consensus is honest in that they only included valid transactions. They therefore do not fully execute the consensus algorithm to know which blocks are accepted, and may end up in a situation where they accept blocks that contain invalid transactions, that full nodes have rejected.

\subsection{Sampling-Based Data Availability}\label{sec:data-availability}

The `data availability problem' asks how a client--such as a light client--that only downloads block headers, but not the corresponding block data (\eg list of transactions), can satisfy itself that the block data is not being withheld by the producer of the block (\eg a miner), and that the full data is indeed available to the network.

Al-Bassam \etal \cite{al-bassam2018} propose a solution to this problem based on erasure coding and random sampling. The solution was proposed in the context of state transition fraud proofs, however it is of independent interest. We summarise the scheme here.

Erasure codes are error-correcting codes \cite{peterson1972} that can transform a message consisting of $k$ shares (\ie pieces) into a bigger extended message of $n$ shares, such that the original message can be recovered by any subset $k'$ of the $n$ shares. The ratio $\frac{k'}{n}$ (the code rate) depends on the erasure code used and its parameters. For example, Reed-Solomon erasure codes \cite{wicker1994} can support $k' = \frac{n}{2}$, which means that only half of the erasure coded data is needed to recover the original data.

To allow clients to be sure that block data is available, block headers contain a commitment to the root of the Merkle tree of the erasure coded version of the data. In order for an adversarial block producer to withhold any part of the block, they must withhold at least $k'$ out of $n$ shares of the block (\eg with standard Reed-Solomon coding this would be at least 50\% of the block). Clients can then sample multiple random shares from the block, and if it does not receive a response for one of its samples because it is unavailable, then it considers the whole block to be unavailable, and does not accept the block. If an adversarial block producer has withheld $k'$ out of $n$ shares, then there is a high probability that the client will sample an unavailable piece and reject the block.

However because the block producer may incorrectly compute the erasure code or the Merkle tree, thus making the block data unrecoverable, it is necessary to allow clients to receive `fraud proofs' from full nodes to alert them that the erasure code is incorrect, causing the client to download the block data, recompute the erasure code and verify that it does not match the Merkle root.

To prevent clients from needing to download the entire block data (which would defeat the goal of data availability proofs being more efficient than downloading the whole data yourself), two-dimensional erasure coding is used, which limits these fraud proofs to a specific axis as only one row or column needs to be downloaded to prove that the erasure code is incorrectly computed, thus the fraud proof size would be approximately $O(\sqrt{n})$ (without Merkle proofs) for a block with $n$ shares. However this also requires clients who want data availability guarantees to download a Merkle root for the shares in each row and column as part of the block header, rather than a single Merkle root for the entire data, thus the number of Merkle roots that need to be downloaded increases from $1$ to $2\sqrt{n}$, as there are $\sqrt{n}$ rows and $\sqrt{n}$ columns.

Importantly, that for this scheme to provide any guarantees, there must be a minimum number of clients in the network that are making enough samples to force the block producer to release more than $k'$ shares to satisfy all of those samples, as if less than $k'$ shares are released, the block data may not be recoverable from those shares. This is because clients gossip downloaded samples to `full' nodes that request it, so that they can recover the full block with enough shares, similar to a peer-to-peer file-sharing network.

    \section{Model}\label{sec:model}

\subsection{Threat Model and Nodes}

We consider the following types of nodes in \sysname:
\begin{itemize}
    \item \textbf{Consensus nodes.} These are nodes which participate in the consensus set, to decide which blocks should be added to the chain.
    \item \textbf{Storage nodes.} These are nodes which store a copy of all of the data in the blockchain and its blocks.
    \item \textbf{Client nodes.} These are effectively the end-users of the blockchain system. They participate in applications or use cases that use the blockchain, and receive transaction data or messages from storage nodes relevant to their applications.
\end{itemize}

These nodes are all connected to each other in a peer-to-peer network, \eg all node types may have some connections with any other node type and the topology of the network is non-hierarchical. However, client nodes may wish to ensure that they are connected to at least one storage node if they wish to utilise their services.

We assume that honest nodes not under an eclipse attack \cite{heilman2015} and are thus connected to at least one other honest node; that is, a node that will follow the protocols described in \Cref{sec:block-validity-rule} and relay messages. This implies that the network is not split, so that there is always a network path between two honest nodes. Additionally, there is at least one honest storage node in the network. \mustafa{Can we relax the assumption that the network is not split? \alberto{Maybe Sabre~\cite{sabre} can help you with that; \eg assume the presence of Sabre nodes.}}

We also assume that there is a maximum network delay $\delta$ so that if an honest node can receive a message from the network at time $T$, then any other honest node can also do so at time $T' \leq T + \delta$. \alberto{Is this the same assumption as FraudProof? If it is, you may mention it to make it clear that the data availability proof from FraudProof apply also here (since they work under the same assumption).}

\subsection{Block Model}

We assume a blockchain data structure that at minimum consists of a hash-based chain of block headers $H = (h_0, h_1, ...)$. Each block header $h_i$ contains the root $\mathsf{mRoot}_i$ of a Merkle tree of a list of messages $M_i = (m_i^0, m_i^1, ...)$, such that given a function $\mathsf{root}(M)$ that returns the Merkle root of a list of messages $M$, then $\mathsf{root}(M_i) = \mathsf{mRoot}_i$. This is not an ordinary Merkle tree, but an ordered Merkle tree we refer to as a `namespaced' Merkle tree which we describe in \Cref{sec:namespaced-merkle-tree}. A block header $h_i$ is considered to be valid if given some boolean function
\begin{equation*}
    \mathsf{blockValid}(h) \in \{\mathsf{true}, \mathsf{false}\}
\end{equation*}
then $\mathsf{blockValid}(h_i)$ must return $\mathsf{true}$.

If a block is valid, then it has the potential to be included in the blockchain. We assume that the blockchain has some consensus rules to decide which valid blocks have consensus to be included in the blockchain, and resolve forks. A block header $h_i$ is considered to have consensus if given some boolean function
\begin{equation*}
    \mathsf{inChain}(h, \{H_0, H_1, ...\}) \in \{\mathsf{true}, \mathsf{false}\}
\end{equation*}
then $\mathsf{inChain}(h_i, \{H_0, H_1, ...\})$ must return $\mathsf{true}$, where each $H_j$ is a distinct chain of block headers and $\{H_0, H_1, ...\}$ is the set of distinct chains observed (there may be multiple in the event of a fork).

Note that computing $\mathsf{inChain}$ on $h_i$ can only return $\mathsf{true}$ if and only if $\mathsf{blockValid}(h_i)$ returns $\mathsf{true}$, regardless of the forks to pick from. \alberto{This means that \texttt{blockValid($h_i$)=$true$} is a necessary condition for having \texttt{inChain($h_i$)=$true$}, but it is not a sufficient condition; right? In that case, should you prove this?} Apart from this constraint, the specific consensus rules used by $\mathsf{inChain}$ are arbitrary and are out of scope for the design of \sysname. For example, $\mathsf{inChain}$ may use proof-of-work or proof-of-stake with the longest chain rule \cite{nakamoto2008, bano2017}.

\subsection{Goals}\label{sec:goals}

With this threat model in mind, \sysname has the following goals:

In the text below, `messages relevant to the application' means messages that are necessary to compute the state of an application, and is discussed in more depth in \Cref{sec:cross-application}.

\begin{enumerate}
    \item \textbf{Availability-only block validity.} The result of $\mathsf{blockValid}(h_i)$ should be $\mathsf{true}$ if the data behind $\mathsf{mRoot}_i$ is available to the network. This consequently means that consensus nodes should not need to execute messages in blocks. \alberto{Just to be sure, would $\mathsf{blockValid}(h_i)$ return $false$ if the data are not available? (or would it return $\perp$?)}
    \item \textbf{Application message retrieval partitioning.} Client nodes must be able to download all of the messages relevant to the applications they use from storage nodes, without needing to downloading any messages for other applications. \alberto{I guess that clients will also have to download some "extra" data in order to validate/find those messages (eg. merkle proof); should you provide an upper bound on the size of these "extra messages"?}
    \item \textbf{Application message retrieval completeness.} When client nodes download messages relevant to the applications they use from storage nodes, they must be able to verify that the messages they received are the complete set of messages relevant to their applications, for specific blocks, and that there are no omitted messages. \alberto{Should client also be able to verify the order of the messages (\ie. nodes cannot trick them into believing that some messages arrived first than others.)? I guess you get that for free since you use Merkle trees.\mustafa{That is correct, maybe I should talk about that, but at the moment I don't consider per-block message ordering.}}
    \item \textbf{Application state sovereignty.} Client nodes must be able to execute all of the messages relevant to the applications they use to compute the state for their applications, without needing to execute messages from other applications, unless other specific applications are explicitly declared as dependencies. \mustafa{Encapsulation? \alberto{Chainspace defines a notion of encapsulation; but is it the same notion? I believe Chainspace says "objects can only be modified by the smart contract that created them" (ie. it relates to write operations); here the notions seems more specific to read operations}}
\end{enumerate}

    \section{Block Validity Rule Design}\label{sec:block-validity-rule}

The key idea of \sysname is that the result of $\mathsf{blockValid}(h_i)$ should only depend on whether the data required to compute $\mathsf{mRoot}_i$ is available to the network or not, rather on whether any of the messages in the block correspond to transactions that satisfy the rules of some state machine (Goal 1 in \Cref{sec:goals}). This way, we can decouple transaction validity rules from the consensus rules, as the result of $\mathsf{inChain}$ does not depend on the contents of the messages in the block $M_i$, when $\mathsf{blockValid}(h_i)$ is computed (recall $\mathsf{inChain}$ on $h_i$ can only return $\mathsf{true}$ if and only if $\mathsf{blockValid}(h_i)$ returns $\mathsf{true}$).

We consider that checking the availability of the data necessary to recompute $\mathsf{mRoot}_i$ is the bare minimum necessary requirement to have a useful functioning blockchain. This is because, as we shall see in \Cref{sec:application-model}, clients need to know the transactions that have occurred in the blockchain in order to know the state of applications on the blockchain and thus do anything useful. If the data behind a block is unavailable, clients would not be able to compute the state of their applications.

We provide definitions for data availability soundness and agreement, adapted from \cite{al-bassam2018} for the threat model described in \Cref{sec:model}.

\begin{definition}[Data Availability Soundness]\label{def:data-availability-soundness}
    If an honest node accepts a block as available, then at least one honest storage node has the full block data or will have the full block data within some known maximum delay $k \times \delta$ where $\delta$ is the maximum network delay.
\end{definition}

\begin{definition}[Data Availability Agreement]\label{def:data-availability-agreement}
    If an honest node accepts a block as available, then all other honest nodes will accept that block as available within some known maximum delay $k \times \delta$ where $\delta$ is the maximum network delay.
\end{definition}

We offer two possible validity rules with different trade-offs. \Cref{sec:simplistic-validity-rule} describes a simple validity rule that satisfies \Cref{def:data-availability-soundness} and \Cref{def:data-availability-agreement} with 100\% probability, for an $O(n)$ bandwidth cost where $n$ is the size of the block, because the node must download the entire block data to confirm that it is available. \Cref{sec:probabilistic-validity-rule} describes a probabilistic validity rule that satisfies \Cref{def:data-availability-soundness} and \Cref{def:data-availability-agreement} with a high but less than 100\% probability, but with a $O(\sqrt{n} + \log(\sqrt{n}))$ bandwidth cost because the block's row and column Merkle roots and only a static number of samples and their logarithmically-sized Merkle proofs from the block need to be downloaded. This bandwidth cost is analysed further in \Cref{sec:implementation}.

\subsection{Simplistic Validity Rule}\label{sec:simplistic-validity-rule}

In the Simplistic Validity Rule, $\mathsf{blockValid}(h_i)$ returns $\mathsf{true}$ if and only if upon receiving a block header $h_i$ from the network, the node is also able to download $M_i$ from the network and authenticate that the Merkle root of the downloaded $M_i$ is $\mathsf{mRoot}_i$, by checking that $\mathsf{root}(M_i) = \mathsf{mRoot}_i$.

Upon $\mathsf{blockValid}(h_i)$ returning $\mathsf{true}$, the node must distribute $h_i$ and $M_i$ to the nodes it is connected to, should the nodes request the data if they do not have it. The node should thus store $M_i$ for at least $\delta$, the maximum network delay.

\begin{theorem}
    The Simplistic Validity Rule satisfies \Cref{def:data-availability-soundness} (Soundness).
\end{theorem}
\begin{proof}
    If $\mathsf{blockValid}(h_i)$ returns $\mathsf{true}$ on an honest node, then the node will distribute $M_i$ to the nodes it is connected to, of at least one of which is honest, and will also run $\mathsf{blockValid}(h_i)$ and distribute $M_i$, and so on. Thus a storage node will receive $M_i$ within the maximum network delay $\delta$, which there exists at least one of which is honest.
\end{proof}

\begin{theorem}
    The Simplistic Validity Rule satisfies \Cref{def:data-availability-agreement} (Agreement).
\end{theorem}
\begin{proof}
    If $\mathsf{blockValid}(h_i)$ returns $\mathsf{true}$ on an honest node, then the node will distribute $M_i$ to the nodes it is connected to, of at least one of which is honest, and will also run $\mathsf{blockValid}(h_i)$ and distribute $M_i$, and so on. Thus all honest nodes will receive $M_i$ within the maximum network delay $\delta$, and $\mathsf{blockValid}(h_i)$ will thus return $\mathsf{true}$, causing them to accept $h_i$ as an available block.
\end{proof}

\subsection{Probabilistic Validity Rule}\label{sec:probabilistic-validity-rule}

For the Probabilistic Validity Rule, $\mathsf{blockValid}(h_i)$ utilises the probabilistic data availability scheme based on random sampling the erasure coded version of the block data $M_i$ described by Al-Bassam \etal \cite{al-bassam2018} and summarised in \Cref{sec:data-availability}. Proofs for \Cref{def:data-availability-soundness} and \Cref{def:data-availability-agreement} are provided in \cite{al-bassam2018}. Unlike the Simplistic Validity Rule, this scheme is probabilistic in satisfying these definitions, however it is more efficient because only a part of the block data needs to be downloaded to obtain high probability guarantees that the data is available.

For examples if using the 2D Reed-Solomon coding scheme with a $\frac{1}{4}$ code rate described in \cite{al-bassam2018} in a block that has been divided into 4096 shares, only 15 samples corresponding to 0.4\% of the block data needs to be downloaded by a node to achieve a 99\% guarantee that the block data is available \cite{al-bassam2018}. Further analysis will be provided in \Cref{sec:implementation}.

The bandwidth cost of executing $\mathsf{blockValid}(h_i)$ is $O(\sqrt{n} + \log(\sqrt{n}))$ where $n$ is the size of the block, because each node needs to download $2\sqrt{n}$ row and column Merkle roots for the block's 2D erasure coded data, and a fixed number of share samples and their corrosponding Merkle proofs authenticating them to one of the block's row or column roots (which are logarithmic in size).

As mentioned in \Cref{sec:background}, this scheme only works if there is a sufficient minimum number of nodes in the network making a sufficient number of sample requests so that the network collectively samples enough shares to be able to reconstruct the block, thus the maximum block size and number of samples each node makes should be set to reasonable values such that this condition is met.

We note that this creates an interesting property: in order to (securely) increase block size and thus the throughput of the network, one can increase the number of nodes in the network. This is different to traditional blockchain systems such as Bitcoin \cite{nakamoto2008}, where deploying more full nodes to the network does not increase the on-chain throughput of the network. By reducing block verification to data availability verification, the blockchain has scalability properties more similar to those of peer-to-peer file-sharing networks \cite{benet14}, where adding more nodes to the network increases the storage capacity of the network. See Table 1 in \cite{al-bassam2018} for example parameterisation and numbers for the minimum nodes that are required for certain block sizes.

Additionally, as also mentioned in \Cref{sec:background}, in the (hopefully rare) case that the erasure code is incorrectly generated, the size of the fraud proof for this would be approximately $O(\sqrt{n})$, or $O(\sqrt{n} + \sqrt{n}\log(\sqrt{n}))$ including Merkle proofs (for example, for a 1MB block with 225 byte shares, the size of a fraud proof would be 26KB; see \cite{al-bassam2018} for more figures).
    \section{Application-Layer Design}\label{sec:application-model}

\subsection{Application Model}

\begin{figure}
    \centering
    \includegraphics[width=\linewidth]{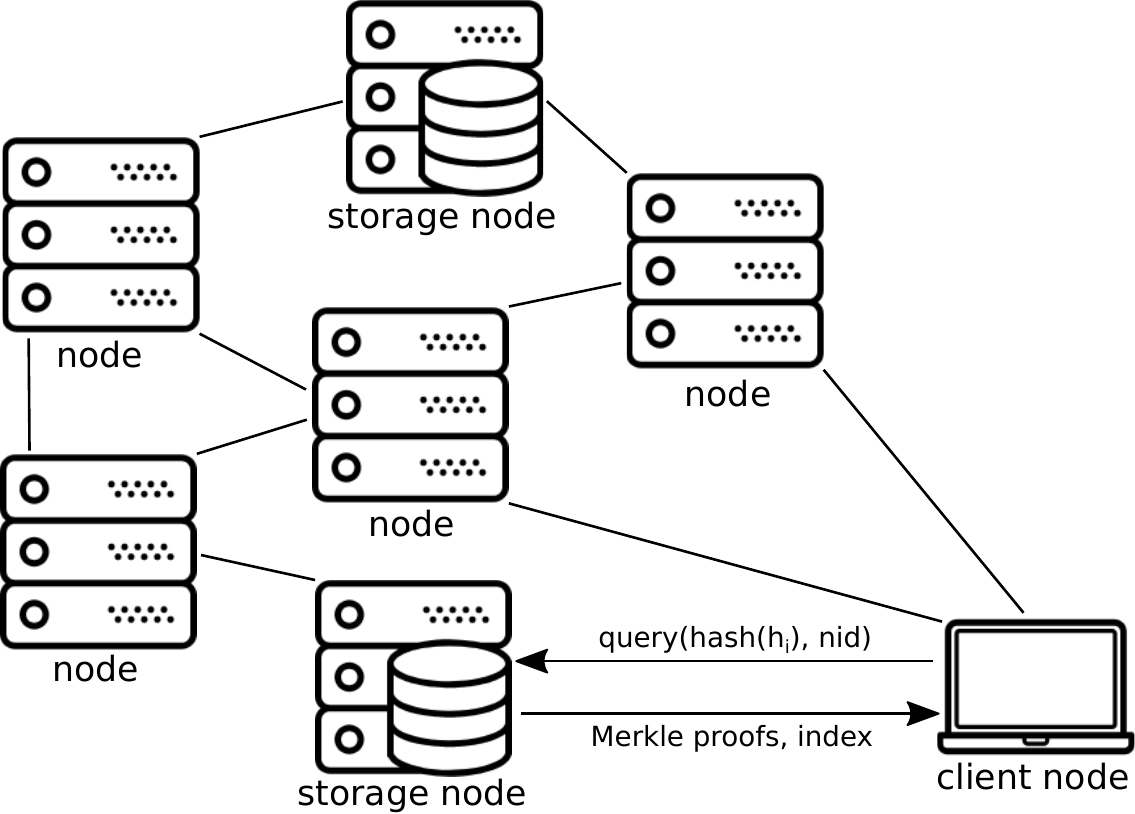}
    \caption{Overview of interaction between client nodes and storage nodes. }
    \label{fig:client-protocol}
\end{figure}

Recall in \Cref{sec:model} that \sysname has client nodes which read and write messages in blocks relevant to their application, and that the contents of blocks have no validity rules, and thus any arbitrary message can be included in a block. \sysname applications are akin to smart contracts, with the primary difference being that they are executed by end-user clients rather than consensus participants. Thus, application logic is defined and agreed upon entirely off-chain by clients of that application, and may therefore be written in any programming language or environment. 

A client can submit a message to be recorded on the blockchain that specifies a transaction for a specific application, which will then be read and parsed by other clients of that applications, which may then modify their copy of the state of that application.

Applications are identified by their own `namespace', and well-formed messages associated with a specific application can be parsed to determine their application namespace. We assume a function $\mathsf{ns}(m)$ that takes as input a message $m$ and returns its namespace ID.  Therefore if a client is a user of an application with ID $\mathsf{nid}$, it is interested in
reading all messages $m$ in the ledger such that $\mathsf{ns}(m) = \mathsf{nid}$, in order to determine the state of its application.
\alberto{What are the characteristics of that function; does it associate \emph{at most} one namespace ID per message $m$? does it associate \emph{exactly} one namespace ID per message $m$? Is this function one-way? etc}

Because the consensus of the blockchain does not require checking the validity of any transactions included in the blockchain, the ledger may include transactions that are considered invalid according to the logic of certain applications. Therefore we define a state transition function that \sysname applications should use that does not return an error. Given an application with ID $\mathsf{nid}$:

\begin{equation*}
    \mathsf{transition}_{\mathsf{nid}}(\mathsf{state}, \mathsf{tx}) = \mathsf{state'}
\end{equation*}

$\mathsf{transition}_{\mathsf{nid}}(\mathsf{state}, \mathsf{tx})$ cannot return an error because if an adversarial actor includes an invalid $\mathsf{tx}$ in a block, then the state of the application would end up in a permanently erroneous state. Therefore if $\mathsf{tx}$ is considered erroneous by the logic of the transition function, it should simply return the previous state, $\mathsf{state}$, as the new state.

Clients who use an application with ID $\mathsf{nid}$ should agree with each other on the logic or code of $\mathsf{transition}_\mathsf{nid}$. If for example, one client decides to use different logic for $\mathsf{transition}_\mathsf{nid}$, then that client would reach a different final $\mathsf{state}$ for that application than everyone else, which in effect means that they would be using a different application, but it would not effect anyone else.

Interestingly, this means that it is possible for users of an application to decide to change the logic of that application without requiring a hard-fork of the blockchain that would effect other applications. However if immutability of the logic is important, the creator of the application may decide for example that the namespace identifier of the application should be the cryptographic hash of the application's logic. \alberto{This seems very important, but buried into the text; What about emphasising it (potentially selling it from  the abstract/intro)?}

\subsubsection{Cross-Application Calls}\label{sec:cross-application}

Some applications may want to call other applications (\ie a cross-contract call). We consider two scenarios in which an application may want to do this: either as a pre-condition or a post-condition. We consider a model where all cross-application calls can be expressed as pre-conditions or post-conditions, similar to \eg the transaction model of Chainspace \cite{al-bassam2018sharding}.

Recall in \Cref{sec:goals} that Goal 4 of \sysname is application state sovereignty, which means that users of an application should not have to execute messages from other irrelevant applications. An application can specify other applications as dependencies in its logic, where knowledge of the state of the dependency applications is necessary in order to compute the state of the application. An application $B$ is thus defined as `relevant' to users of application $A$ if is $B$ is a dependency of $A$, however if $A$ is not a dependency of $B$, then $A$ is not relevant to the users of $B$. In order to preserve the notion of state sovereignty, this means third party applications cannot force other applications to take a dependency on the state of third party applications.

In the case of a pre-condition, an application may have a function that can only be executed if another application that it depends on is in a certain state. In such a case, in order to validate that these pre-conditions are met, clients of an application must also download and verify the state of the application's dependency applications; however the clients of the dependency application do not need to download the state of the applications which depend on it.

For example, consider a name registrar application where clients can register names only if they send money to a certain address in a different currency application. The clients of the name registrar application would have to also become clients of the currency application, in order to verify that when a name is registered, there is a corresponding transaction that sends the funds to pay for the name to the correct address.

In the case of a post-condition, an application may want to modify the state of another application after a transaction. Post-conditions are only possible if the application whose state is being modified has explicitly set the application that is executing the post-condition as a dependency application to the post-condition application. This is because in order to execute the post-condition, the clients of the post-condition application would have to download and verify the state of the application executing the post-condition, to verify that it has the authority to execute the post-condition. If any application was allowed to execute a post-condition in any application, then it would mean that clients would have to download and verify other applications against their will, thus violating Goal 4 in \Cref{sec:goals} (application state sovereignty). Post-conditions may however be executed indirectly through sidechain mechanisms such as federated pegs \cite{back2014}, but this is out of scope for this paper.

\subsection{Storage Nodes and Namespaced Merkle Tree}\label{sec:namespaced-merkle-tree}

\begin{figure}
    \centering
    \includegraphics[width=\linewidth]{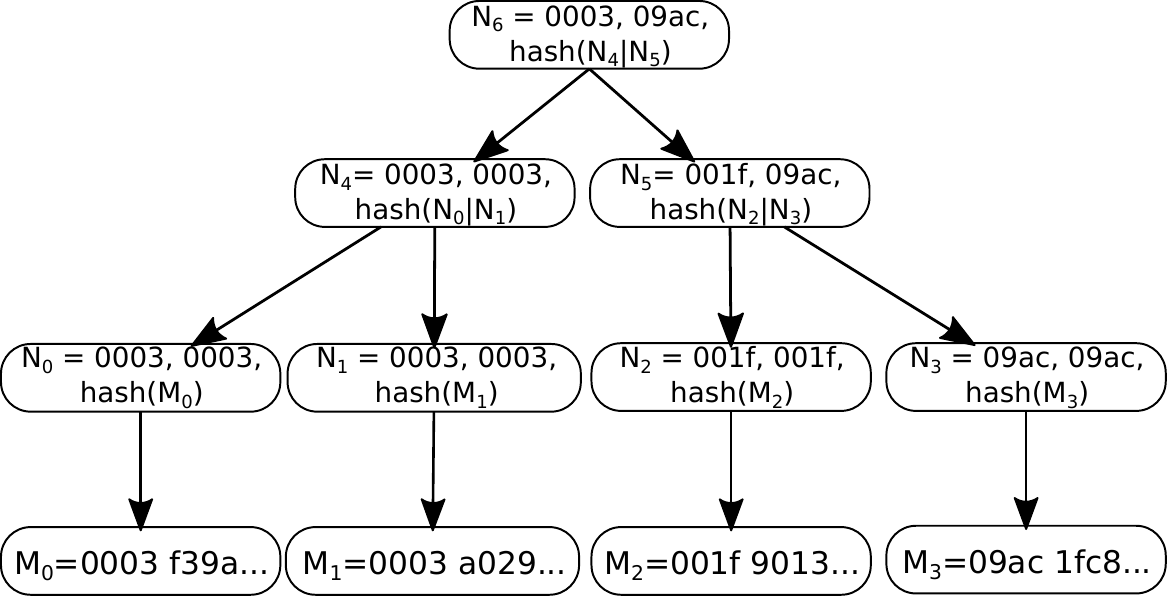}
    \caption{An example of a namespaced Merkle tree.}
    \label{fig:namespaced-merkle-tree}
\end{figure}

In order to satisfy Goal 2 in \Cref{sec:goals} (application message retrieval partitioning) to allow client nodes to be able to retrieve all the messages relevant to the application namespaces they are interested in without having to download and parse the entire blockchain themselves (\eg if they use the Probabilistic Validity Rule, or simply assume that the consensus has a honest-majority that only accepts available blocks), they may query storage nodes for all of the messages in a particular application namespace for particular blocks. The storage node can then return Merkle proofs the relevant messages being included in the blocks.

In order to allow storage node to prove to clients that they have returned the complete set of messages for a namespace included in a block's Merkle tree of messages (Goal 3 in \Cref{sec:goals}, application message retrieval completeness), we use a `namespaced` Merkle tree described below, which is an ordered Merkle tree that uses a modified hash function so that each node in the tree includes the range of namespaces of the messages in all of the descendants of each node. The leafs in the tree are ordered by the namespace identifiers of the messages.

In a namespaced Merkle tree, each non-leaf node in the tree contains the lowest and highest namespace identifiers found in all the leaf nodes that are descendants of the non-leaf node, in addition to the hash of the concatenation of the children of the node. This enables Merkle inclusion proofs to be created that prove to a verifier that all the elements of the tree for a specific namespace have been included in a Merkle inclusion proof.

The Merkle tree can be implemented using standard unmodified Merkle tree algorithms, but with a modified hash algorithm that depends on an existing hash function, that prefixes hashes with namespace identifiers. Suppose $\mathsf{hash}(x)$ is a cryptographically secure hash function such as SHA-256. We define a wrapper function $\mathsf{nsHash}(x)$ that produces hashes prefixed with namespace identifiers. A namespaced hash has the format $\mathsf{minNs} || \mathsf{maxNs} || \mathsf{hash(x)}$, where $\mathsf{minNs}$ is the lowest namespace identifier found in all the children of the node that the hash represents, and $\mathsf{maxNs}$ is the highest.

The value of $\mathsf{minNs}$ and $\mathsf{maxNs}$ in the output of $\mathsf{nsHash}(x)$ depends on if the input $x$ is a leaf or two concatenated tree nodes, as illustrated by \Cref{fig:namespaced-merkle-tree}. If $x$ is a leaf, then $\mathsf{minNs} = \mathsf{maxNs} = \mathsf{ns}(x)$, as the hash contains only one leaf with a single namespace.

If $x$ is two concatenated tree nodes, then $x = \mathsf{left}||\mathsf{right}$ where $\mathsf{left} = \mathsf{leftMinNs} || \mathsf{leftMaxNs} || \mathsf{hash(x)}$ and $\mathsf{right} = \mathsf{rightMinNs} || \mathsf{rightMaxNs} || \mathsf{hash(x)}$. Thus in the output of $\mathsf{nsHash}(x)$, $\mathsf{minNs} = \mathsf{min}(\mathsf{leftMinNs}, \mathsf{rightMinNs})$ and $\mathsf{maxNs} = \mathsf{max}(\mathsf{leftMaxNs}, \mathsf{rightMaxNs})$.
\alberto{The above equation really helped me to understand; but it is buried in text (and a bit ugly).}

An adversarial consensus node may attempt to produce a block that contains a Merkle tree with children that are not ordered correctly. To prevent this, we can set a condition in $\mathsf{nsHash}$ such that there is no valid hash when $\mathsf{leftMaxNs} \geq \mathsf{rightMinNs}$, and thus there would be no valid Merkle root for incorrectly ordered children. Therefore $\mathsf{blockValid}(h_i)$ would return false in the simplistic and probabilistic validity rules as there is no possible $M_i$ where $\mathsf{root}(M_i) = \mathsf{mRoot}_i$. Additionally, recall that $\mathsf{root}(M_i) = \mathsf{mRoot}_i$ and thus $\mathsf{blockValid}(h_i)$ would also return false if the Merkle root of the tree is constructed incorrectly, \eg if the minimum and maximum namespaces for a node in the tree are incorrectly labelled.

Because only the hash function is being modified in the Merkle tree, the Merkle tree is generated, and Merkle proofs are verified using standard algorithms. However, during Merkle proof verification, and extra step is necessary in order to verify that the proofs covers all of the messages for a specific namespace.

A client node can send a query $\mathsf{query}(\mathsf{hash}(h_i), \mathsf{nid})$ to a storage node to request all of the messages in block $h_i$ that have namespace ID $\mathsf{nid}$. The storage node replies with a list of Merkle proofs $\mathsf{proofs} = (\mathsf{proof}_0, \mathsf{proof}_1, ..., \mathsf{proof}_n)$ and an index $\mathsf{index}$ that specifies the index in the tree in which $\mathsf{proof}_0$ is located. In addition to the client node verifying all the proofs, the client node also verifies that the highest namespace in all of the left siblings included in $\mathsf{proof}_0$ are smaller than $\mathsf{nid}$, and the lowest namespace in all of the right siblings included in $\mathsf{proof}_n$ are larger than $\mathsf{nid}$.

If a block has no messages associated with $\mathsf{nid}$, then only one proof $\mathsf{proof}_0$ is returned which corresponds to the child in the tree where the child to the left of it is smaller than $\mathsf{nid}$ but the child to the right of it is larger than $\mathsf{nid}$. The actual message in the child does not need to be included in the proof as the purpose of the proof would just be to show that there are no messages in the tree for $\mathsf{nid}$.

\begin{theorem}
    Assuming the Merkle tree is correctly constructed, an incomplete set of Merkle proofs $\mathsf{proofs} = (\mathsf{proof}_0, \mathsf{proof}_1, ..., \mathsf{proof}_n)$ for a request for the messages of $\mathsf{nid}$ can always be detected.
\end{theorem}
\begin{proof}
    Let us assume that an adversary returns an incomplete set of correct proofs $\mathsf{proofs} = (\mathsf{proof}_0, \mathsf{proof}_1, ..., \mathsf{proof}_n)$ for $\mathsf{nid}$, and $\mathsf{index}$ is the index in the tree that $\mathsf{proof}_0$ is located at.
    
    If an omitted message for $\mathsf{nid}$ has an index lower than $\mathsf{index}$, then $\mathsf{proof}_0$ will contain a left sibling node with a maximum namespace $\mathsf{maxNs}$ where $\mathsf{maxNs} > \mathsf{nid}$, thus proving that there is an omitted message to the left of the proof set.
    
    If an omitted message for $\mathsf{nid}$ has an index higher than $\mathsf{index} + n$, then $\mathsf{proof}_n$ will contain a right sibling node with a minimum namespace $\mathsf{minNs}$ where $\mathsf{nid} > \mathsf{minNs}$, thus proving that there is an omitted message to the right of the proof set.
\end{proof}

\subsection{DoS-resistance}

In the design of \sysname, consensus nodes are not responsible for validating transactions, and thus an adversarial client may submit many invalid transactions for namespaces, forcing clients to download many invalid transactions. In a permissioned system, consensus nodes can choose which clients can submit transactions. However in a permissionless system, there ought to be a way to prioritise transactions and to make it expensive to conduct DoS attacks.

\subsubsection{Transaction Fees}

Consensus nodes can choose to prioritise transactions that include transaction fees. However, any transaction fee system should ideally not require client nodes that read messages from the application namespaces they are interested in, to also validate the application that implements the currency that transaction fees are paid in.

To achieve this, when a message is submitted to consensus nodes for inclusion in a block, the submitter of the message can also submit a `fee transaction' for the currency application, and also attach to the fee transaction the hash of the `child' message that the fee is paying for, such that the fee in this special fee transaction can only be collected if the message behind the specified hash is included in the same block, according to the logic of the currency application.

Client nodes of the original application whose message that the fee is paying for do not need to validate the fee transactions in the currency application; only client nodes of the currency application (\eg the consensus nodes) do. Additionally, the client nodes of the currency system application do not have to download the child message itself to verify that it has been included in the block and thus the fee has been earned, but simply verify a Merkle proof that the hash of the child message is included in the same block.

We assume that fee transactions only specify one dependency message for simplicity, but they may specify multiple dependency messages.

There does not need to be a native currency to the system, as consensus nodes can choose to accept transaction fees in any currency application that they choose to recognise. \alberto{This seems an important insight; you may want to emphasise it.}

\subsubsection{Maximum Block Size}

A maximum block size can be implemented without requiring nodes to download the entire block's data to verify that it is below a certain size. Instead, each message, and thus each leaf in the Merkle tree of messages, may have a maximum size such that if a message $x$ is bigger than the allowed size, $\mathsf{nsHash}(x)$ would return an error, so $\mathsf{root}(M_i) = \mathsf{mRoot}_i$ and thus $\mathsf{blockValid}(h_i)$ would return false. If larger message sizes are required, a message could be chunked into multiple messages and parsed back into a single message by clients.
    \section{Implementation and Performance}\label{sec:implementation}

We implemented a prototype of \sysname in 2,865 lines of Golang code. The code has been released as a free and open-source project.\footnote{\url{https://github.com/musalbas/lazyledger-prototype}}

As well as the core \sysname system, we also implemented (and released) several example applications using \sysname. Each application's state is implemented as one or more key-value stores that can be read from or modified. Applications include:
\begin{itemize}
    \item A currency application where clients publish messages that are transactions for the transfer of funds between addresses that are elliptic curve public keys. Transactions are signed by the public keys of senders, and specify the amount of funds to send and the recipient address. In the key-value store, keys are public keys, and values are the corresponding balance of each public key, which is updated after each valid transaction.
    \item A name registration application where clients can: \first send a balance top-up transaction to the registrar's public key using a dependency currency application, so that clients can pay for name registrations using their balance with the registrar; and \second send a registration transaction to register a specified name to their public key, which reduces the balance of the registrant, if their balance is sufficient. The registration application has one key-value store representing the in-app topped-up balance of each public key, and another key-value store where each key represents a registered name and each value represents the public key the name has been registered to.
    \item A dummy application for testing purposes which adds arbitrary sized specified key-value pairs to its key-value store.
\end{itemize}

We present an evaluation of \sysname's performance and scalability properties.

\Cref{fig:plot1} compares how much data needs to be downloaded to execute the Simple Validity Rule and the Probabilistic Validity Rule to verify data availability, for varying block sizes. As expected, there is a linear relationship between block size and data downloaded for the Simple Validity Rule, as this requires downloading all of the block data to ensure that it is available. However, we can see that the relationship between block size and data downloaded for the Probabilistic Validity Rule is sub-linear and almost flat. This is because in order to execute the Probabilistic Validity Rule, nodes download a fixed number of samples and their corresponding Merkle proofs whose sizes increase logarithmically with the size of the block, as well as set of $2\sqrt{n}$ row and column Merkle roots for the block where $n$ is the size of the block.

\Cref{fig:plot2} compares the response size of queries for a specific namespace to a storage node (``application proofs''), for varying amounts of messages of different namespaces (measured by total bytes) that are not relevant to that query. We use currency application messages as the relevant queried namespace (although any other application could be used), fixing the number of currency messages in the block to 10, but increasing the total size of dummy application messages. We can observe that that for both simple blocks and probabilistic blocks, the size of the application proofs for the relevant application only increases logarithmically, because although messages that are not in the relevant namespace do not need to be downloaded, the size of the Merkle proofs for those messages increase logarithmically as the number of total messages in the block increases. The size of the application proofs for probabilistic blocks are smaller because a two-dimensional erasure code is used, where each column and row gets its own Merkle tree, and thus the Merkle proofs are smaller because there are less items in each tree.

\Cref{fig:plot3} follows the same setup as \Cref{fig:plot2}, however instead of comparing the size of the applications proofs for the currency application, we compare the size of the state that needs to be stored by users of the relevant application (in this case, the currency application). As expected, we observe that as the size of the state of other applications increase, the size of the state that needs to be stored for the currency application remains static.

\Cref{fig:plot4} and \Cref{fig:plot5} illustrate how the size of application proofs may vary for an application that has a dependency application. In this case we use the name registration application as an example, which requires users to follow the state of a currency application so that balance top-up transactions to the registrar can be verified. In the two graphs, we setup multiple instances of the name registration application for multiple registrars, but the user is only interested in following one of them. In \Cref{fig:plot4} we can observe that as the number of top-up transactions for the irrelevant name registration applications increase, the size of the application proofs for the relevant name registration application increases linearly, because the user must also download application proofs for the currency application, which has transactions being added to it by users of the other name registration applications. This extreme case where there are only top-up transactions defeats any scalability gains of \sysname, since all transactions require transactions in dependency applications that other users may follow.

However, \Cref{fig:plot5} shows the same but in the case of name registration transactions instead of balance top-up transactions. Here we see that irrelevant name registration transactions do not linearly increase the size of application proofs that need to be downloaded for other users, because only users of the relevant name registration application need to have knowledge of the registered names, and no dependency application is impacted.
\alberto{Would be better to increase the font of labels and axes of the figures.}

\begin{figure}
    \centering
    \includegraphics[width=\linewidth]{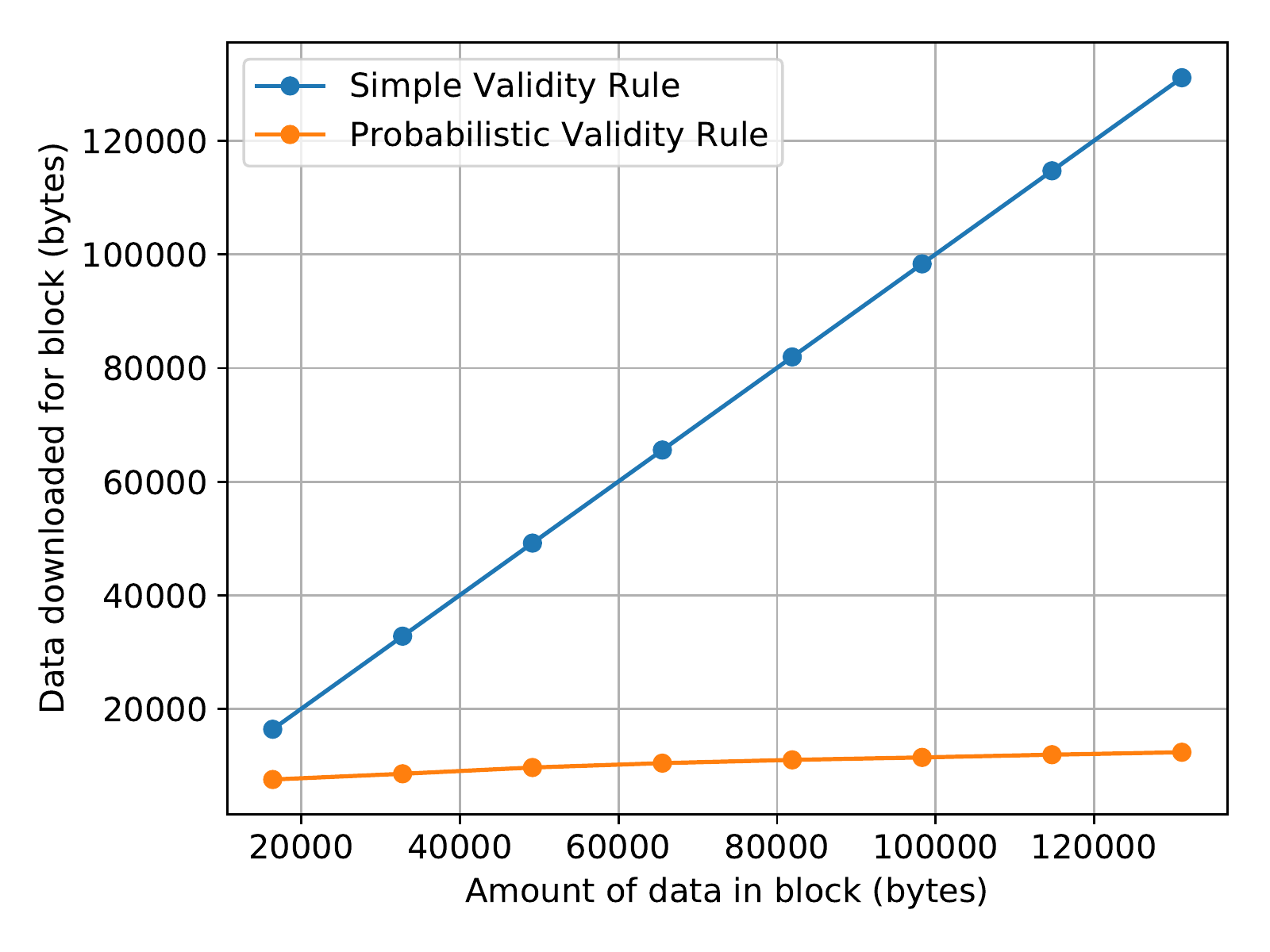}
    \caption{A graph showing how much data needs to be downloaded to execute the block validity rule to validate data availability versus the size of the block. For the Probabilistic Validity Rule, 15 samples are used.}
    \label{fig:plot1}
\end{figure}

\begin{figure}
    \centering
    \includegraphics[width=\linewidth]{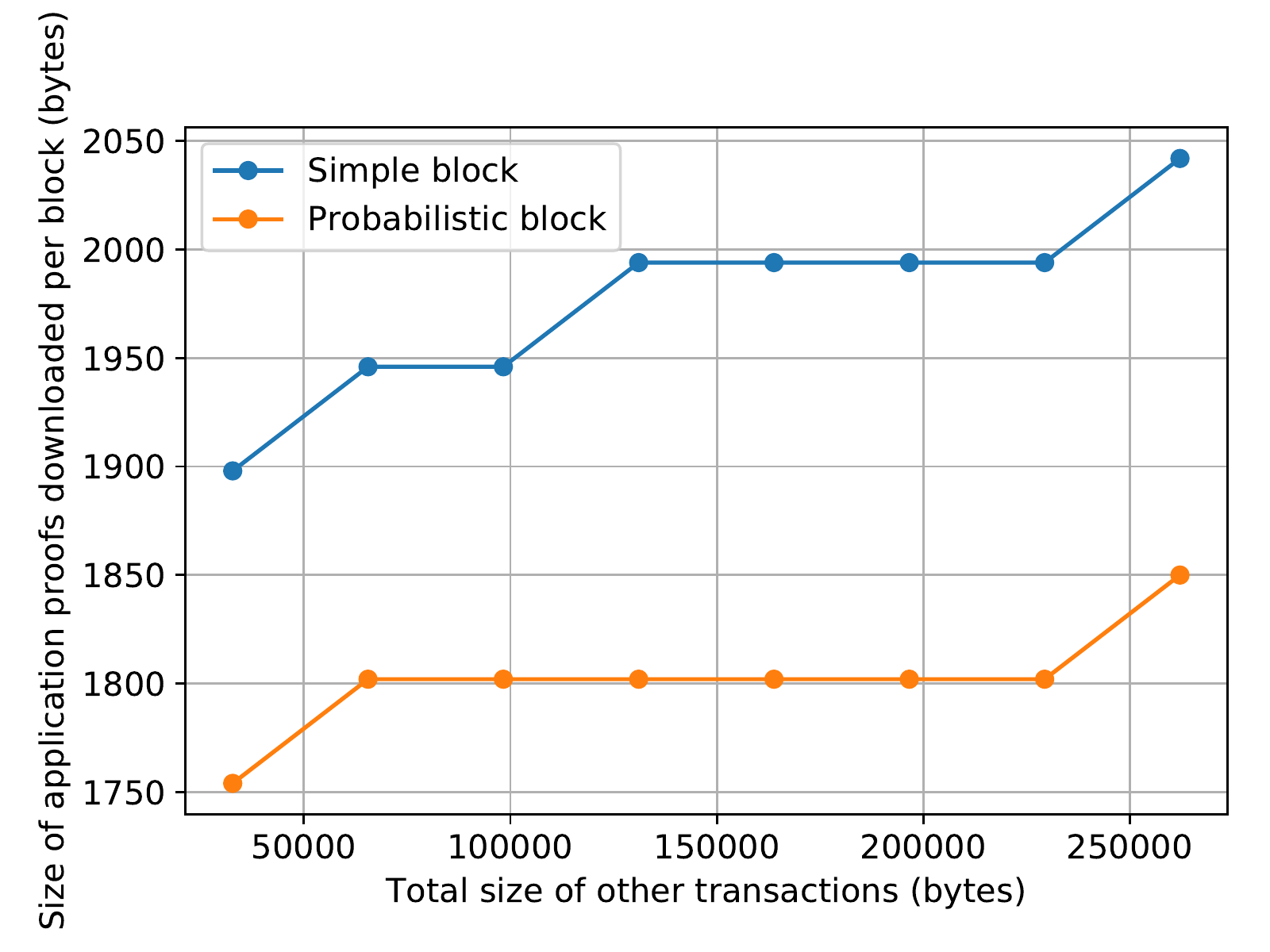}
    \caption{A graph showing the size of application proofs in a block for a currency application with 10 transactions versus the total size of all of the other transactions in the block.}
    \label{fig:plot2}
\end{figure}

\begin{figure}
    \centering
    \includegraphics[width=\linewidth]{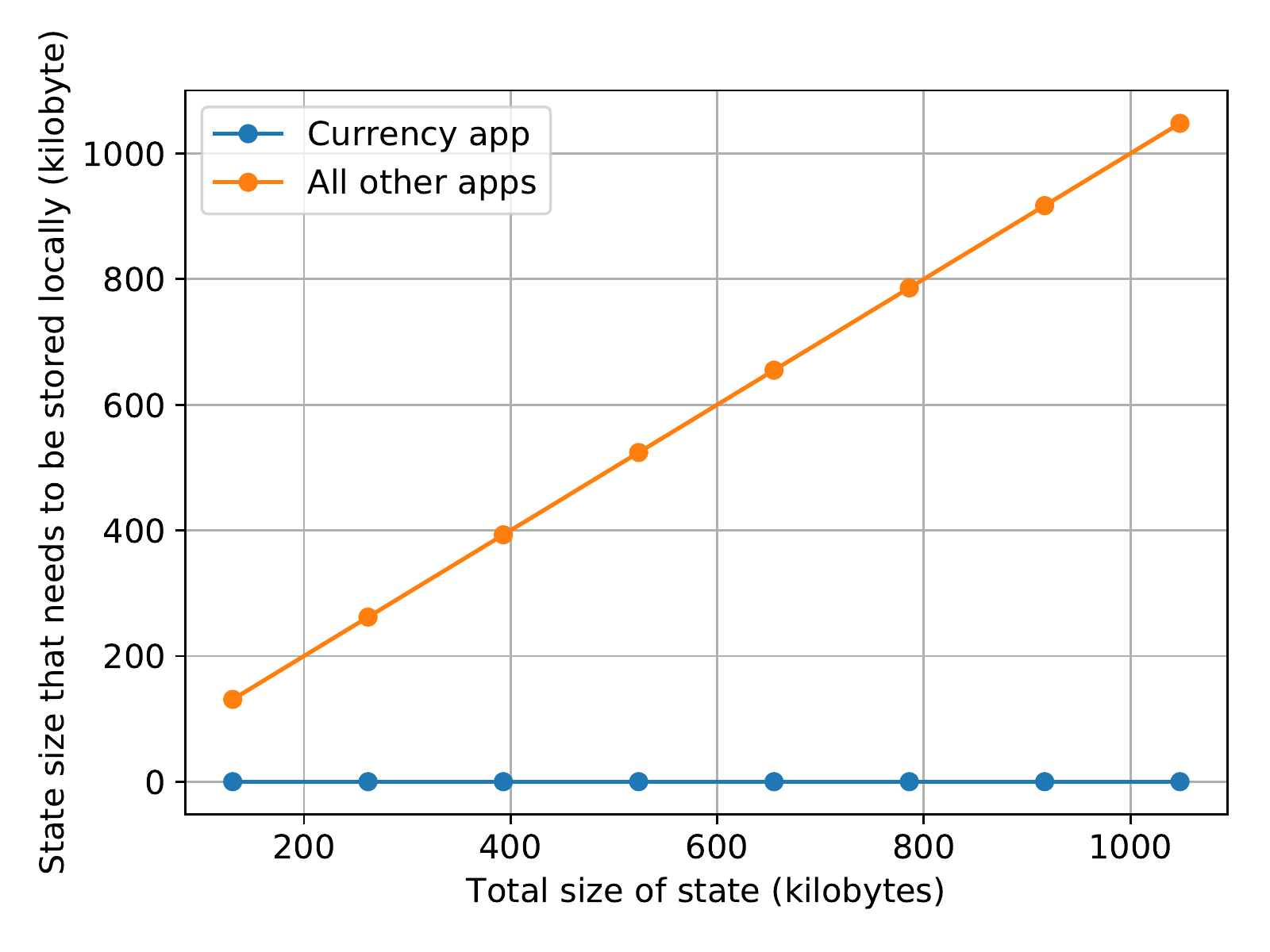}
    \caption{A graph showing the size of the state that needs to be stored after a block versus the total size the state of all apps in the block, for a currency app and all other apps.}
    \label{fig:plot3}
\end{figure}

\begin{figure}
    \centering
    \includegraphics[width=\linewidth]{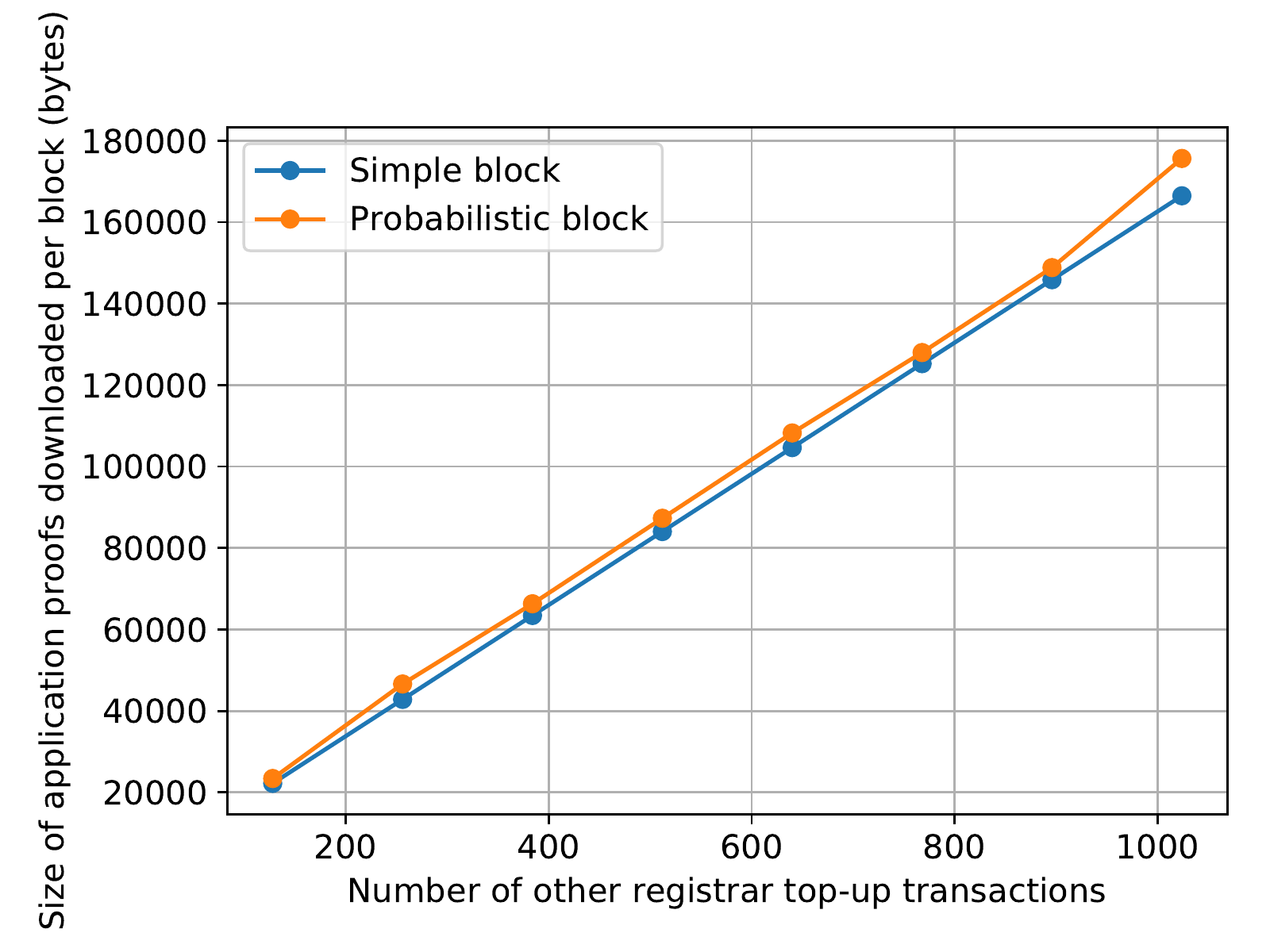}
    \caption{A graph showing the size of application proofs in a block for an instance of a registrar application with 10 top-up transactions versus the number of top-up transactions for other registrar application instances in the block.}
    \label{fig:plot4}
\end{figure}

\begin{figure}
    \centering
    \includegraphics[width=\linewidth]{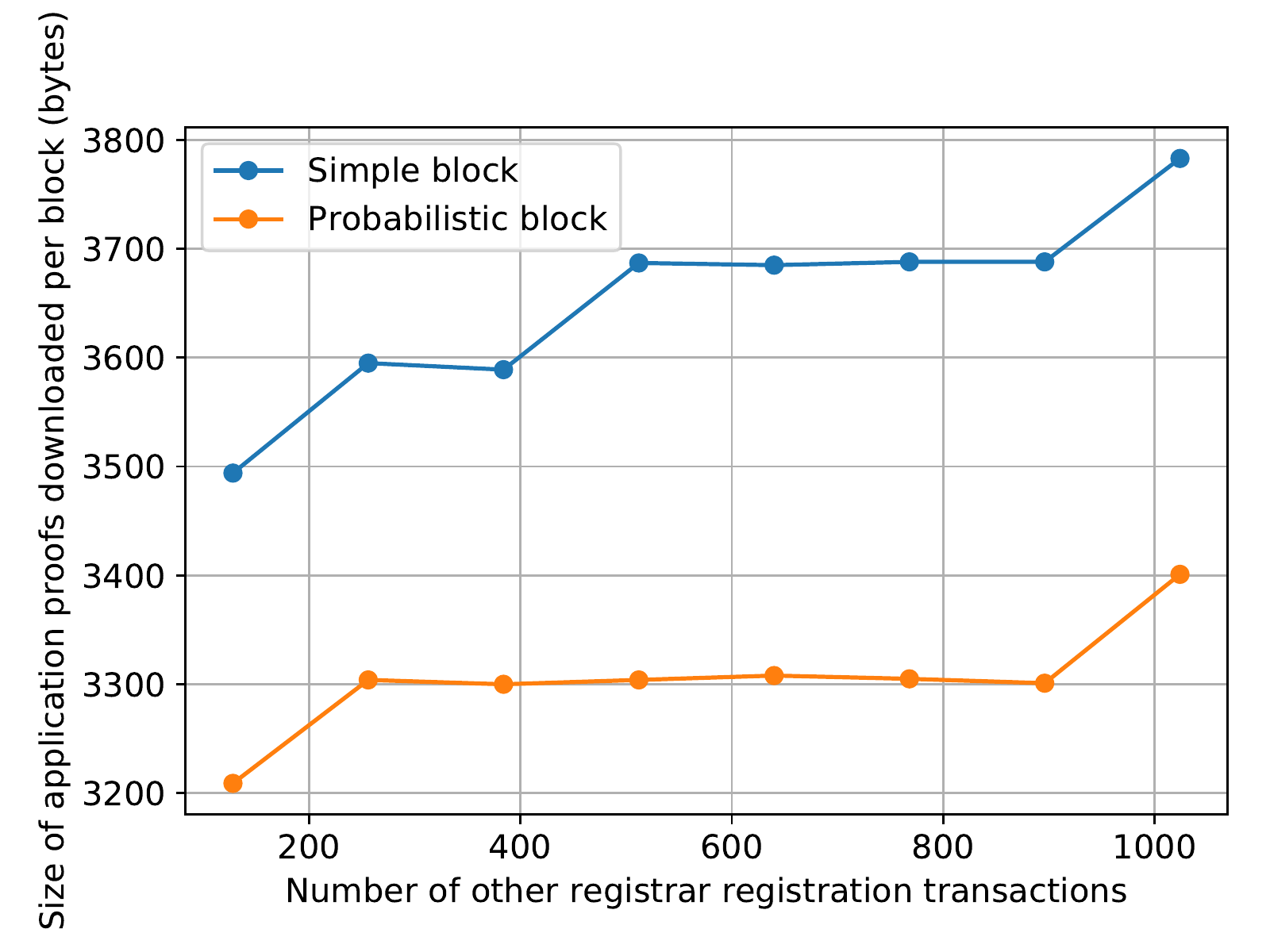}
    \caption{A graph showing the size of application proofs in a block for an instance of a registrar application with 10 registration transactions versus the number of registration transactions for other registrar application instances in the block.}
    \label{fig:plot5}
\end{figure}

\mustafa{Same graphs for throughput?}
    \section{Discussion}

\subsection*{Application Light Clients}

One of the current limitations of \sysname is that it is not obvious how to build light clients for applications, so that clients do not have to download all of an application's messages to learn the application's state. This is because application messages are not validated by consensus nodes, thus clients cannot assume that an honest majority has validated them. This may be an interesting area of future work. \alberto{I may be wrong, but it seems to me that only "generic" light-clients seem impossible (\ie light-client working for any kind of applications). But we may perhaps design \sysname light-client that work for specific applications. For example, I may use \sysname for a particular application specifying that a particular path of the namespace tree contains an accumulator of the latest state of my application; could I design a light-client for my particular application (knowing that all my clients know that particular path)? It seems to me that in \sysname, it is up the application-designer to provide a client software (since no-one else know how to interpret the application data); it might therefore not be impossible for the application-designer to design a light-client (depending on its applicaiton).\mustafa{Yeah true, I shouldn't say it's impossible, I changed it.}}
\alberto{Perhaps, other options include applications where a set of authorities (\eg the application-designers) threshold sign checkpoints, and include them on a particular address-space. Then light-clients may simply sync by downloading the latest checkpoint. (This implies a threshold assumption, which is publicly verifiable by non-light clients).}

\subsection*{Hard Forking}

One of the interesting aspects of the \sysname design is its consequences on blockchain governance, in particular hard forks. Traditionally, hard forks have been used in the past to change transaction protocol rules \cite{larson2017} or to reverse damage caused popular by smart contracts being compromised, such as the DAO hack \cite{buterin2016}.

However with \sysname, as there are no transaction-specific protocol rules and blocks may contain any arbitrary data, hard forks to change transaction rules or change the state of the system are not possible or necessary, as the interpretation of transactions are left to the device of the end-user clients rather than the consensus. Thus if users of a specific application decide they want to change the state of or `upgrade' an application, they can do so without the permission of the consensus or any on-chain effects or changes, as long as other users implement the same upgrade. Users who do not implement the upgrade will locally interpret the application to have a different state - similar to the effect of a hard-fork but without requiring one explicitly.
    \section{Related Work}\label{sec:related-work}

The namespaced Merkle tree in \sysname is inspired by the `flagged' Merkle tree concept by Crosby and Wallach \cite{crosby2009}, where each node in the tree is has a flag that represents the attributes that its leafs has.

Mastercoin (now OmniLayer) \cite{willett2012} is a blockchain application system predating Ethereum \cite{buterin2013}, which uses Bitcoin has a protocol layer for posting messages. This is similar to \sysname in the sense that the blockchain can be used to post arbitrary messages that are interpreted by clients, however in Mastercoin all nodes must download all Mastercoin messages as the Bitcoin base layer does not support efficient data availability schemes such as the Probabilistic Validity Rule. Additionally, as Mastercoin uses Bitcoin as the base layer, it does not support queries for complete sets of messages by specific applications by clients. Finally, Mastercoin has a set of hardcoded applications, and does not support arbitrary applications. In contrast, \sysname examines what an ideal new blockchain would look like for use as a base layer in a system where the base layer is only for posting messages and data availability.

\alberto{It seems that one of the novelty of \sysname is the way it decouples transaction execution from transaction ordering. To emphasise that, you may want to add this section a comparison with other systems that aim to achieve this goal (\eg. \cite{vukolic2017rethinking}, but I think there are other BFT designs that have a similar goal); a good selling point is that \sysname is different since the validation is entirely done by the clients. Also, \sysname verification is "sharded" on the clients since each client only validates a subset of transactions (this seems also novel).}
    \section{Conclusion}\label{sec:conclusion}

We have presented and evaluated \sysname, a unique blockchain design paradigm where the base layer is only used a mechanism to guarantee the availability of on-chain messages, and transactions are interpreted and executed by end-users. We have shown that by reducing block verification to data availability verification, blocks can be verified in sub-linear time. Additionally, using the notion of application state sovereignty, we have shown that multiple sovereign applications can use the same chain for data availability, with only limited impact to the workload of each other's users.
    \section*{Acknowledgements}
Mustafa Al-Bassam is supported by a scholarship from The Alan Turing Institute.
    
Thanks to George Danezis for helpful discussions about the design of \sysname, and Alberto Sonnino for comments.

    \bibliography{references}
    \bibliographystyle{ieeetr}

    \appendix
\end{document}